\documentclass[%
 reprint,
 amsmath,amssymb,
 aps,
]{revtex4-1}

\usepackage{algorithm}

\usepackage{algpseudocode}

\usepackage{amsthm}
\usepackage{amsfonts}
\usepackage{bbm}
\usepackage{graphicx}
\usepackage{epstopdf}
\usepackage{dcolumn}% Align table columns on decimal point
\usepackage{bm}% bold math
\usepackage{braket}
\usepackage{color}
\usepackage{enumerate}
\usepackage{booktabs}

\theoremstyle{definition}
\theoremstyle{plain}\newtheorem{theorem}{Theorem}
\theoremstyle{definition}\newtheorem{remark}{Remark}
\theoremstyle{definition}
\theoremstyle{plain}\newtheorem{proposition}[theorem]{Proposition}
\theoremstyle{plain}\newtheorem{corollary}{Corollary}
\theoremstyle{plain}\newtheorem{lemma}{Lemma}
\theoremstyle{plain}

\makeatletter
\newcommand{\noindentalgorithmic}[1]{%
  \setlength{\ALG@thistlm}{0pt}%
  \begin{algorithmic}[1]
    #1
  \end{algorithmic}%
}
\makeatother

\begin{document}

\preprint{APS/123-QED}
%\linenumbers

%\title{Distributed Phase Estimation Algorithm }

\title{Universal Error Correction for Distributed Quantum Computing  }\thanks{$\dag$ issqdw@mail.sysu.edu.cn (Corresponding author's address)}
%\thanks{A footnote to the article title}%

%\iffalse
%\fi
\author{Daowen Qiu$^{1,\dag}$, Ligang Xiao$^{1}$, Le Luo$^{2}$, Paulo Mateus$^{3}$}
%\email{$\dag$ issqdw@mail.sysu.edu.cn (Corresponding author's address)}
\affiliation{
$^1$  School of Computer Science and Engineering, Sun Yat-sen University, Guangzhou 510006, China \\
%$^2$ The Guangdong Key Laboratory of Information Security Technology, Sun Yat-sen University, 510006, China\\
$^2$ School of Physics and Astronomy, Sun Yat-sen University, Zhuhai 519082, China \\
$^3$ Instituto de Telecomunica\c{c}\~{o}es, Departamento de Matem\'{a}tica,
Instituto Superior T\'{e}cnico,  Av. Rovisco Pais 1049-001  Lisbon, Portugal
}

%$\dag$912574652@qq.com;\\

%\fi
%\author{Daowen Qiu$\dag$, Ligang Xiao}
%\email{$\dag$ issqdw@mail.sysu.edu.cn (D. Qiu, Corresponding author's address)}
%\affiliation{
 %School of Computer Science and Engineering, Sun Yat-sen University, Guangzhou 510006, China \\}
% The Guangdong Key Laboratory of Information Security Technology, Sun Yat-sen University, 510006, China\\}

%\date{\today}

\begin{abstract}
In   distributed quantum computing, the final solution of a problem is usually achieved by catenating these partial solutions resulted from different computing nodes, but   intolerable errors likely yield in this catenation process. In this paper, we propose a universal error correction scheme  to reduce errors and obtain effective solutions. Then, we apply this error correction scheme to designing a distributed phase estimation algorithm that presents  a basic tool for studying distributed Shor's algorithm and distributed discrete logarithm algorithm as well as other distributed quantum algorithms. Our method may provide a universal strategy of error correction for a kind of distributed quantum computing.

%represented by a 0-1 bit string

 %Shor's algorithm is one of the most significant quantum algorithms. Shor's algorithm can factor large integers with a certain success probability in polynomial time. However, Shor's algorithm requires an unbearable amount of qubits in the NISQ (Noisy Intermediate-scale Quantum) era. To reduce the resources required for Shor's algorithm, in this paper we first propose a new distributed phase estimation algorithm. Our distributed phase estimation algorithm does not require quantum communication and it reduces the number of qubits of a single node compared to the traditional phase estimation algorithm (non-iterative version). Then we apply our distributed phase estimation algorithm to form a distributed order-finding algorithm for Shor's algorithm. Compared with the traditional Shor's algorithm (non-iterative version), the maximum number of qubits required by a single node of our distributed order-finding algorithm is reduced by  $(2-\dfrac{2}{k})L-\log_2k-O(1)$ when factoring an $L$-bit integer ($k$ is the number of compute nodes). The communication complexity of our distributed order-finding algorithm is $O(kL)$. 

\newtheorem{defi}{Definition}

%\begin{description}
%\item[Usage]
%Secondary publications and information retrieval purposes.
%\item[PACS numbers]
%May be entered using the \verb+\pacs{#1}+ command.
%\item[Structure]
%You may use the \texttt{description} environment to structure your abstract;
%use the optional argument of the \verb+\item+ command to give the category of each item.
%\end{description}
\end{abstract}

\pacs{Valid PACS appear here}% PACS, the Physics and Astronomy
                             % Classification Scheme.
%\keywords{Suggested keywords}%Use showkeys class option if keyword
                              %display desired
\maketitle

%\tableofcontents

\section{Introduction}\label{sec:introduction}

Quantum computing has been rapidly developing with impressive advantages over classical computing. However, in order to realize quantum algorithms in practice, medium or large scale general quantum computers are required. Currently it is still  difficult to implement such quantum computers. Therefore,  to advance the application of quantum algorithms in the NISQ era, we would consider to reduce the required qubits or other quantum resources for quantum computers.
	
	Distributed quantum computing is a computing method of combining distributed computing and quantum computing, and has been significantly studied (for example,[1-8]  and references therein). Its purpose is to solve problems by means of fusing multiple smaller quantum computers working together. Distributed quantum computing is usually used to reduce the resources required by each computer, including qubits, gate complexity, circuit depth and so on. Due to these potential benefits, distributed quantum algorithms have been designed in recent years  [9-21].
For example, in 2013, Beals et al. proposed an algorithm for parallel addressing quantum memory \cite{BBG2013}.	In 2018, Le Gall et al. studied quantum algorithms in the quantum CONGEST model \cite{LM2018}. 
In 2022, Qiu et al. proposed a distributed Grover's algorithm \cite{QLX2024}, and Tan et al. proposed a distributed quantum algorithm for Simon's problem \cite{TXQ2022}. There are many important contributions concerning distributed quantum computing and algorithms, but here we do not expound the details in these references [1-21].  In general, these distributed quantum algorithms can reduce quantum resources to some extent.
	
	%Shor's algorithm \cite{shor1994algorithms}  is  one of the most significant algorithms in quantum computing. It can factor large integers with a certain probability of success and costs polynomial time. Since the best known classical algorithm for factoring large numbers is subexponential but superpolynomial, Shor's algorithm demonstrates quantum advantages. Shor's algorithm can be applied to break RSA encryption which has been widely used in public key cryptography system. Shor's algorithm can be implemented in two ways: one needs to measure multiple qubits at the end (we call it \textit{non-iterative Shor's algorithm}, e.g. \cite{kaye2006introduction,nielsen2000quantum,shor1994algorithms}), and the other alternately performs unitary operators and measurements, and only  one qubit is measured at a time (we call it \textit{iterative Shor's algorithm}, e.g. \cite{beauregard2003circuit,haner2017factoring,parker2000efficient}). The iterative Shor's algorithm has only one control qubit and it requires $2L+O(1)$ qubits when factoring an $L$-bit integer \cite{beauregard2003circuit,haner2017factoring,parker2000efficient}. The non-iterative Shor's algorithm has $2L+O(1)$ control qubits and thus it requires $4L+O(1)$ qubits \cite{kaye2006introduction,nielsen2000quantum,shor1994algorithms}.

	If  a result outputted by a quantum computer for solving a problem is described by a bits string, then we may consider to use $k$-computing nodes (smaller scale) to get  $k$ substrings, respectively,  and by catenating these substrings we may obtain an appropriate solution for the original problem. However, if the procedure of catenating these substrings is straightforward without processing error correction appropriately, then it  usually leads to intolerable errors. Therefore, the aim in the paper is to analyze and establish a universal error correction scheme for a kind of distributed quantum computing.

Phase estimation algorithm plays an important role in Shor's algorithm (\cite{NC2000, Shor1994}), and other quantum algorithms \cite{HHL2009}. As an application of the proposed error correction algorithm, we  design a distributed phase estimation algorithm,  with the advantages of less qubits and quantum gates over centralized one. The designed distributed phase estimation algorithm likely provides a basic tool for further studying other distributed quantum algorithms, for example, distributed Shor's algorithm and distributed discrete logarithm algorithm as well as distributed HHL algorithm \cite{HHL2009}.

	The remainder of the paper is organized as follows. First, in Section II, we propose a kind of problems concerning distributed quantum computing and two potential schemes for error correction. Then in Section III, we present and prove the useful error correction scheme in detail. As an application, Section IV serves to apply the error correction scheme in Section III to designing a distributed phase estimation algorithm.  Finally, in Section V, we summarize the main results and mention potential problems for further study. 
	
%\vskip -20mm

\section{Formulation of Problems and Schemes}\label{SFPM}

In this paper, for any 0-1 string $x$, we use $l[x]$ to denote the length of  $x$, and use  $d[x]$ to denote its decimal number;  on the other hand, if $x$ is a decimal number, then we use $b[x]$ to represent its binary number correspondingly. % In addition, we use $l[x]$ to denote the length of  $x$.

For any $x\in\{0,1\}^n$ ($n\geq k$),  we use $P[x,k]$ and $S[x,k]$  to denote the prefix and suffix of $x$ with $k$ bits, respectively. %and use $M[x,i,j]$ to denote the substring of $x$ from $i$the to $j$th. 
For example, if $x=01100110$, then $l[x]=8$, $d[x]=102$, $P[x,3]=011$, $S[x,3]=110$. %$M[x,3,6]=1001$.

Let  $\{0,1\}^n$ be a distance space by defining its distance  $D_n$ as: for any $x,y\in\{0,1\}^n$, 
\begin{equation}
D_n(x,y)=\min\left(|d[x]-d[y]|, 2^n-|d[x]-d[y]|\right).
\end{equation}

Actually,   $D_n$ can be verified to satisfy the conditions as a distance later on.

%(see \cite{gang2023distributed}).

Suppose that the solution or approximate solution of a problem can be described by a string $\omega=a_1a_2\cdots a_n\in\{0,1\}^n$, and an output  $\beta=b_1b_2\ldots b_n\in\{0,1\}^n$  from a quantum algorithm solving this problem satisfies 
\begin{equation}D_n(\omega,\beta)\leq 1.\end{equation} 
Then, by means of distributed quantum computing to solve this problem, we can consider two scenarios in the following. 

{\bf Scheme 1:}
We may divide $\omega$ into $k$ substrings, say $A_1=a_1a_2\cdots a_{i_1}$,  $A_2=a_{i_1+1}a_{i_1+2}\cdots a_{i_2}$, $\ldots$, $A_k=a_{i_{k-1}+1}a_{i_{k-1}+2}\cdots a_n$. That is, $a=A_1\circ A_2\circ \cdots \circ A_k$, where $\circ$ denotes ``catenation" operation, but we often omit $\circ$ and write $a=A_1A_2 \cdots A_k$ simply, if no confusion results. Then we
use $k$ computing nodes to estimate the $k$ substrings, and obtain $k$ substrings $S_1,S_2,\ldots, S_k$ with the same length as $A_1, A_2,\ldots,A_k$,  respectively,  satisfying 
\begin{equation}
D_{l[A_i]}(A_i,S_i)\leq 1,
\end{equation}
for $i=1,2,\ldots,k$.
% $s=s_1s_2\ldots s_k\in\{0,1\}^n$ such that $D_n(a(s_i),s_i)\leq 1$,  where $a(s_i)$ denotes the  substring of $a$  corresponding to the same position of $s_i$ in $s$, $i=1,2,\ldots,k$.
 However, unfortunately, by using this scheme, we can not ensure that 
  \begin{equation}
D_n(\omega,S)\leq 1,
\end{equation}
  where $S=S_1S_2\ldots S_k$, but this is required  for an approximate solution. So, this scheme is straight but not feasible in general (e.g. \cite{LQL2017}). 

{\bf Scheme 2:} We also divide the  $\omega=a_1a_2\cdots a_n\in\{0,1\}^n$ into $k$ substrings,  say $A_1,A_2,\ldots,A_k$,  but these substrings are not the same as the above {\bf Scheme 1}, because we require there are certain overlaps between adjacent substrings. More specifically, 
 for a given $k_0$ (it is not longer than the length of each substring),  the suffix of $A_i$ with length $k_0$ is overlapped with the prefix of $A_{i+1}$ with the same length $k_0$,  that is, $S[A_i, k_0]=P[A_{i+1},k_0]$,
 $i=1,2,\ldots,k-1$. We use $k$ computing nodes denoted as $Q_1,\cdots,Q_k$ to estimate  $A_1,A_2,\ldots, A_k$, respectively.
 
 If the output of $Q_i$ is $S_i$, then it is also required to satisfy 
 \begin{equation}\label{LE}
D_{l[A_i]}(A_i,S_i)\leq 1,
\end{equation}
for $i=1,2,\ldots,k$. 
 
 Finally, in order to get the solution from $S_1,S_2,\ldots, S_k$,  we can utilize the bits with overlapped positions to correct $S_{k-1}, S_{k-2},\dots,S_1$ in sequence  one by one,  and by combining all corrected substrings appropriately we can achieve the solution $S'$ satisfying
  \begin{equation} \label{CR}
D_n(A,S')\leq 1,
\end{equation}
 where the last substring $S_k$ is fixed without being changed and it is used for correcting $S_{k-1}$ to obtain a new substring. With this new substring, $S_{k-2}$ is going to be corrected. In sequence, all substrings will be corrected, resulting in a final solution.  Of course, this technical process needs to be strictly formulated and proved in detail in the next section.

 \section{Error Correction} \label{SEC}
 
 % We continue to solve the error correction problem raised  above.  The basic ideas can be divided into the following steps: 
 
In this section, we continue to solve the error correction problem  raised in Section \ref{SFPM}. We first describe the ideas of designing this algorithm, and then give related notations and properties of bit Strings. Finally, we present the concrete algorithm and proof.

 \subsection{Basic Ideas}
 The basic ideas can be divided into the following steps: 
 
 $\textbf{Step 1: }$ Fix $S_k$ and select an element $c_1\in \{\pm 2,\pm 1,0\}$ such that the suffix of $S_{k-1}$ with length $k_0$ (without loss of generality, take $k_0=3$ in this paper) adding $c_1$ equals to the prefix of $S_k$ with length $3$. Then $c_1$ adding $S_{k-1}$ obtains a new substring with the same length as  $S_{k-1}$, and this new substring being catenated with
 the suffix of $S_k$ with length $l[S_k]-3$ 
 yields a new substring denoted by $S_{k-1}'$. (By the way, $k_0=2$ is not enough and it is easy to construct a counterexample.)

  $\textbf{Step 2: } $ Proceed to  select an element $c_2\in \{\pm 2,\pm 1,0\}$ such that the suffix of $S_{k-2}$ with length $3$ adding $c_2$ equals to the prefix of $S_{k-1}'$ with length $3$.  Then $c_2$ adding $S_{k-2}$ obtains a new string with the same length as  $S_{k-2}$, and this new substring being catenated with 
 the suffix of $S_{k-1}'$ with length $l[S_{k-1}']-3$ 
 yields a new substring denoted by $S_{k-2}'$.

  $\textbf{Step 3: } $ Continue to correct $S_{k-3},   S_{k-4}, \ldots,   S_{1}$ step-by step in this way,  and $S_{k-3}',   S_{k-4}', \ldots,   S_{1}'$  are obtained correspondingly, where $S_{1}'$ is exactly the final goal string $S'$ and satisfies Ineq. (\ref{CR}).

  The above steps can be formulated by an algorithm of error correction, but before presenting this algorithm we still need a number of notations and properties to prove the correction of algorithm.
  
\subsection{Notations and Properties of Bit Strings}
  
 For any positive integer $n$, we define an add operation ``$\boldsymbol{+}_n$" as follows.  %``$\boldsymbol{+}$" as follows.    $\oplus$    %on  $\{0,1\}^n\times \{0,\pm1,\pm2,\ldots,\pm 2^n-1\}$ as
\begin{equation}
\boldsymbol{+}_n: \{0,1\}^n\times \{0,\pm1,\pm2,\ldots,\pm (2^n-1)\}\rightarrow \{0,1\}^n
\end{equation}
is defined as: 

 for any $(x,c)\in\{0,1\}^n\times \{0,\pm1,\pm2,\ldots,\pm 2^n-1\}$,
\begin{equation}
x\boldsymbol{+}_n c= b[(d[x]+c)\mod 2^n]
\end{equation}
where  the length of $b[(d[x]+c)\mod 2^n]$ is $n$ by adding some $0$ as its prefix if any.

  Concerning the above distance and operation ``$\boldsymbol{+}_n$", we have the following properties that are useful in this paper.% (the proof is in Appendixes). %The detailed proof can be referred to  \cite{gang2023distributed}.

\begin{proposition}\label{d_t}
For any $x,y\in\{0,1\}^n$, we have: \\
{\rm (I)} $D_n(x,y)=\min\{|b|: x\boldsymbol{+}_n b=y\}$.\\
{\rm (II)} $D_n(x,y)$ is a distance on $\{0,1\}^n$.\\
{\rm (III)} For any $1\leq n_0<n$, if $D_n(x,y)<2^{n-n_0}$, then
\begin{equation}\label{d_t 1}
D_{n_0}(P[x,n_0], P[y, n_0])\leq 1.
\end{equation}
{\rm (IV)} For any $1\leq n_0<n$, if $D_n(x,y)\leq 1$, then
\begin{equation}\label{}
D_{n_0}(S[x,n_0], S[y, n_0])\leq 1.
\end{equation}

\end{proposition}
  
  \begin{proof}

{\rm (I)} 
 Suppose $d[x]\leq d[y]$ and $\left |d[x]-d[y]\right |\leq 2^{n-1}$. Then $D_n(x,y)=d[y]-d[x]$ and $x\boldsymbol{+}_n (d[y]-d[x])=y$.   It is easy to check $d[y]-d[x]=\min\{|b|: x\boldsymbol{+}_n b=y\}$. The other cases are similar.%In this case,  any $c\in \{0,\pm1,\pm2,\ldots,\pm(2^n-1)\}$ with $\boldsymbol{+}_n c=y$ satisfies $c=$

{\rm (II)} We only prove the triangle inequality. For any $x,y,z\in \{0,1\}^n$, denote $D_n(x,y)=|b|$, $D_n(x,z)=|b_1|$, $D_n(z,y)=|b_2|$. Then $x\boldsymbol{+}_n b_1=z$ and $z\boldsymbol{+}_n b_2=y$, and the two equalities result in $x\boldsymbol{+}_n (b_1+b_2)=y$. Consequently, $D_n(x,y)=|b|\leq  |b_1+b_2|\leq |b_1|+|b_2|$.
%Also,  we suppose  $d[x]\leq d[y]$ and $\left |d[x]-d[y]\right |\leq 2^{n-1}$. Then $D_n(x,y)=d[y]-d[x]$

{\rm (III)} From the condition it follows that there exists $b$ with $|b|<2^{n-n_0}$ such that $d[x]+b=d[y]+k2^n$ for some integer $k$.  It is easy to check that  two sides  module  with  $2^{n-n_0}$ leads to 
\begin{equation}
d[P[x,n_0]]+b'=d[P[y, n_0]]
\end{equation}
for some $b'$ with $|b'|<2$, and consequently, we have $D_{n_0}(P[x,n_0], P[y, n_0])\leq |b'|\leq1$.

{\rm (IV)} 
From the condition it follows that there exists $b$ with $|b|\leq 1$ such that $d[x]+b=d[y]+k2^n$ for some integer $k$.  It is easy to check that  two sides  module  with  $2^{n_0}$ leads to 
\begin{equation}
(d[S[x,n_0]]+b)\mod 2^{n_0}=d[S[y, n_0]].
\end{equation}
So, $D_{n_0}(S[x,n_0], S[y, n_0])\leq 1$.

  \end{proof}
  
  \subsection{Error Correction Algorithm}

 In fact,  in {\bf Scheme 2}, without loss of generality, we can also take  $l[A_i]=N_0\geq 3=k_0$ for $i=1,2,\ldots, k-1$, and $N_0\geq l[A_k]=n-kN_0+(k-1)k_0$,  where $k$ depends on $n$. %In the interest of simplicity, more specially, we can directly take $k_0= 3$. %
 %Before presenting the error correction, we describe the ideas of designing this algorithm. 
In order to show the reason of error correction elements $c_i\in\{\pm2,\pm1,0\}$, we need a corollary.

 \begin{corollary} Let 0-1 strings $A_i$ satisfy  $l[A_i]\geq 3$ for $i=1,2,\ldots, k$ and $S[A_i, 3]=P[A_{i+1},3]$,
 $i=1,2,\ldots,k-1$.   Suppose that 0-1 strings $S_i$ satisfy Ineq. (\ref{LE}), i.e., 
  \begin{equation}
D_{l[A_i]}(A_i,S_i)\leq 1,
\end{equation}
for $i=1,2,\ldots,k$. 
Then 
 \begin{equation}
D_{3}(S[S_i,3],P[S_{i+1},3])\leq 2,
\end{equation}
for $i=1,2,\ldots,k-1$. \end{corollary}
 
 \begin{proof}
  
 By means of Proposition \ref{d_t} (III,IV), we have  
  \begin{equation}
D_{3}(S[S_i,3],S[A_{i},3])\leq 1,
\end{equation}
  \begin{equation}
D_{3}(P[A_{i+1},3],P[S_{i+1},3])\leq 1.
\end{equation}
 Due to Proposition \ref{d_t} (II)  (i.e. triangle inequality) and $S[A_i, 3]=P[A_{i+1},3]$, the corollary holds.
 
  \end{proof}

  %and $k$ satisfies $0<n-[(k-1)N_0-3(k-2)]\leq N_0$, that is to say, $0<l[A_k]\leq N_0$.
  Let $\omega\in\{0,1\}^n$ be divided into $k$ bit strings $A_i$ as Scheme 2. We present  the following error correction algorithm ({\bf Algorithm 1}). %For inputting bit strings $S_1,\cdots,S_k$, with $l[S_i]\geq 3$, $i=1,2,\cdots,k$, and  $D_3(S[S_i,3], P[S_{i+1},3])\leq 2$ for $i=1,2,\cdots,k-1$, the following procedure is performed:    

  \textbf{Input}: Bit strings $S_1,\cdots,S_k$, $A_1,\cdots,A_k$, with $l[S_i]=l[A_i]=N_0\geq 3$, $D_{l[A_i]}(A_i,S_i)\leq 1$,  $i=1,2,\cdots,k$, and  $S[A_i, 3]=P[A_{i+1},3]$, for $i=1,2,\cdots,k-1$.    
  
   \textbf{Output}: $S'$ satisfies $D_n(S',\omega)=D_{l[S_k]}(S_k,A_k)$.    
%\textbf{Procedure}:
\begin{algorithmic}[1]
\State Set $S_k'=S_k$.
\For{$r=k-1$ to $1$}
	\State Select $c_r\in\{\pm 2,\pm 1,0\}$ such that
 $S[S_r,3]\boldsymbol{+}_3 c_r=P[S_{r+1}',3]$.  %$ADD\left((m_r)_{[l_{r+1},l_{r+1}+2]},CorrectionNum_r\right)=$$(m_{r+1}')_{[1,3]}$
	%\State %$prefix_r\leftarrow ADD(m_r,CorrectionNum_r)$
	\State $S_r'\leftarrow (S_r\boldsymbol{+}_{l[S_r]}c_r) \circ S[S_{r+1}',l[S_{r+1}']-3]$ (``$\circ$" represents catenation, as mentioned above).  %prefix_r\circ (m_{r+1}')_{[4,len(m_{r+1}')]}$ (``$\circ$" represents catenation)
\EndFor

\State Return $S_1'=S'$.
  
  \end{algorithmic}

In the light  of the above algorithm, %it can obtain such an $S'$ satisfying the requirement of Ineq. (\ref{CR}). More exactly, $D_n(S',\omega)=D_{l[S_k]}(S_k,A_k)$. 
we have the following theorem.% (the proof is in Appendixes).

 % \begin{figure*}
%  \begin{minipage}{\linewidth}

%\begin{algorithm}[H]\label{CC}
%\caption{Error Correction}

%\textbf{Input}: Bit strings $S_1,\cdots,S_k$, with $l[S_i]\geq 3$, $i=1,2,\cdots,k$, and  $D_3(S[S_i,3], P[S_{i+1},3])\leq 2$ for $i=1,2,\cdots,k-1$.    \\
%\textbf{Procedure}:
%\begin{algorithmic}[1]
%\State Set $S_k'=S_k$.
%\For{$r=k-1$ to $1$}
	%\State Choose $c_r\in\{\pm 2,\pm 1,0\}$ such that
 %$S[S_r,3]\boldsymbol{+}_3 c_r=P[S_{r+1}',3]$  %$ADD\left((m_r)_{[l_{r+1},l_{r+1}+2]},CorrectionNum_r\right)=$$(m_{r+1}')_{[1,3]}$
	%\State %$prefix_r\leftarrow ADD(m_r,CorrectionNum_r)$
	%\State $S_r'\leftarrow (S_r\boldsymbol{+}_{l[S_r]}c_r) \circ S[S_{r+1}',l[S_{r+1}']-3]$ (``$\circ$" represents catenation)  %prefix_r\circ (m_{r+1}')_{[4,len(m_{r+1}')]}$ (``$\circ$" represents catenation)
%\EndFor

%\State Return $S_1'=S'$.
%\end{algorithmic}
%\end{algorithm}

 % \end{minipage}
  %\end{figure*}

 \begin{theorem}\label{ERT}
Let  $\omega=a_1a_2\ldots a_n\in\{0,1\}^n$ be divided into $k$ substrings $A_1,A_2,\ldots,A_k$ in turn, with  $l[A_i]=N_0\geq 3$ for $i=1,2,\ldots,k$, and the suffix of $A_i$ with length 3  is overlapped with the prefix of $A_{i+1}$ (i.e., $S[A_i, 3]=P[A_{i+1},3]$), for $i=1,2,\ldots,k-1$.  Suppose that 0-1 strings $S_i$ satisfy Ineq. (\ref{LE}), i.e., 
  \begin{equation}
D_{l[A_i]}(S_i,A_i)\leq 1,
\end{equation}
for $i=1,2,\ldots,k$. 
Then {\bf Algorithm 1} %\ref{CC} 
outputs $S'$ satisfying 
 \begin{equation}
 D_n(S',\omega)=D_{l[S_k]}(S_k,A_k)\leq 1.
 \end{equation}

\end{theorem}

In order to prove the theorem, we first need a lemma as follows.

\begin{lemma}\label{CorrectStep}
Let $A,S$ be two $t$-bit strings ($t\geq 3$). Let $y$ be a $3$-bit string. Suppose $D_t(S,A)\leq 1$ and $D_3(y,S[A,3])\leq 1$. Then:

 (1)  There is unique $b_0$ satisfying $S\boldsymbol{+}_t b_0=A$; and for any $b\in\{\pm 1,0\}$, $S\boldsymbol{+}_t b=A$ if and only if  $S[S,t_0]\boldsymbol{+}_{t_0} b=S[A,t_0]$, where $t_0\leq t$.
 
 (2) There exists unique  $b\in \{\pm 2,\pm 1,0\}$ such that \begin{equation} S[S,3]\boldsymbol{+}_3 b=y.\end{equation}
 
(3) If $S\boldsymbol{+}_tb_1=A$ and $S[A,3]\boldsymbol{+}_3b_2=y$ for some $b_1,b_2\in\{\pm1,0\}$, then  
\begin{equation}
b=b_1+b_2.
\end{equation}
\end{lemma}

\begin{proof}

(1) The proofs are directly derived from  Proposition \ref{d_t} (I).

 (2) From  $D_t(S,A)\leq 1$ it follows that 
 \begin{equation}
 D_3(S[S,3], S[A,3])\leq 1, 
 \end{equation}
  and
\begin{align}
&D_3(S[S,3],y)\\
\leq& D_3(S[S,3], S[A,3])+D_3(S[A,3], y) \leq 2. 
\end{align}
By Proposition \ref{d_t} (I), it holds that  such a $b$ is unique. 

(3) 
It is easy to check that 
\begin{align}
S\boldsymbol{+}_t(b_1+b_2)&=(S\boldsymbol{+}_t b_1)\boldsymbol{+}_t b_2\\
&= A\boldsymbol{+}_t b_2.
\end{align}
So, we have
\begin{align}
S[S,3]\boldsymbol{+}_3 (b_1+b_2)&=S[A,3]\boldsymbol{+}_3 b_2\\
&=y.
\end{align}
Consequently, $b=b_1+b_2$, and the lemma is proved.
\end{proof}

Now we are ready to prove Theorem \ref{ERT}.

Due to $D_{l[S_i]}(S_i,A_{i})\leq 1$, by Proposition \ref{d_t} we have 
\begin{equation} \label{B1}
D_{3}\left( P[S_i,3], P[A_i,3]\right)\leq 1
\end{equation}
and
\begin{equation}\label{B2}
D_{3}\left( S[S_i,3], S[A_i,3]\right)\leq 1,
\end{equation}
$i=1,2,\cdots,k$. 
In the light of Eqs. (\ref{B1},\ref{B2}), we can check that there are $b_1,b_2,b_3\in\{\pm1,0\}$ satisfying
\begin{equation}\label{b1}
S[S_k, 3]\boldsymbol{+}_3 b_1= S[A_k,3],
\end{equation}
\begin{equation}\label{b2}
P[S_k, 3]\boldsymbol{+}_3 b_2= P[A_k,3]=S[A_{k-1},3],
\end{equation}
\begin{equation}\label{b3}
S[S_{k-1}, 3]\boldsymbol{+}_3 b_3= P[A_k,3]=S[A_{k-1},3].
\end{equation}

%Suppose we input $x_1,\cdots,x_k$ to Algorithm \ref{alg:CorrectAndCombine}. Let $prefix_r, m_r'$, $CorrectionNum_r$ ($r=1,\cdots,k-1$) be the same as those in Algorithm \ref{alg:CorrectAndCombine}.

 By virtue of  Lemma \ref{CorrectStep} (2,3), we can further verify that 
\begin{equation}
S[S_{k-1}, 3]\boldsymbol{+}_3 (b_3-b_2)= P[S_k,3].
\end{equation}

Next we need to prove that
\begin{equation}
D_{l[S_{k-1}']} \left(S_{k-1}',A_{k-1}'\right)
=D_{l[S_{k}']}\left(S_{k}',A_k' \right),
\end{equation}
where 
\begin{equation}
S_{k-1}'=(S_{k-1}\boldsymbol{+}_{N_0}(b_3-b_2))\circ S[S_k,l[S_k]-3], 
\end{equation}
\begin{equation}
A_{k-1}'=A_{k-1}\circ S[A_k,l[A_k]-3], S_k'=S_k, A_k'=A_k.
\end{equation}.
 We notify that 
\begin{equation} \label{PSk}
S[S_{k-1}',3]=P[S_k,3].
\end{equation}

First, by virtue of Lemma \ref{CorrectStep} (1) and Eqs. (\ref{b1}), we have 
\begin{equation}\label{Dk}
D_{l[S_{k}']}\left(S_{k}',A_k' \right)=|b_1|.
\end{equation}

Then, also  by virtue of Lemma \ref{CorrectStep} (1) and Eqs. (\ref{B1},\ref{B2},\ref{b1},\ref{b2},\ref{b3},\ref{PSk}), we have

\begin{align}
&S_{k-1}'\boldsymbol{+}_{l[S_{k-1}']} b_1\\
=&P\left[\left(S_{k-1}\boldsymbol{+}_{N_0} (b_3-b_2)\right), N_0-3\right]\circ S_k \boldsymbol{+}_{l[S_{k-1}']} b_1\\
=&P\left[\left(S_{k-1}\boldsymbol{+}_{N_0} (b_3-b_2)\right), N_0-3\right]\circ A_k\\
=&P\left[\left(S_{k-1}\boldsymbol{+}_{N_0} (b_3-b_2)\right), N_0-3\right]\\
&\circ P[A_k,3]\circ S[A_k,l[A_k]-3]\\
=&(P\left[\left(S_{k-1}\boldsymbol{+}_{N_0} (b_3-b_2)\right), N_0-3\right]\\
&\circ  P[S_k,3]\boldsymbol{+}_{N_0}b_2) \circ S[A_k,l[A_k]-3]\\
=&\left(\left(S_{k-1}\boldsymbol{+}_{N_0} (b_3-b_2)\right)\boldsymbol{+}_{N_0}b_2\right)\\
&\circ S[A_k,l[A_k]-3]\\
=&\left(S_{k-1}\boldsymbol{+}_{N_0} b_3\right)\circ S[A_k,l[A_k]-3]\\
=&A_{k-1}\circ S[A_k,l[A_k]-3]\\
=&A_{k-1}',
\end{align}
where we use $S\left[\left(S_{k-1}\boldsymbol{+}_{N_0} (b_3-b_2)\right), 3\right]=P[S_k,3]$.

As a result, we have
\begin{equation}\label{Dk-1}
D_{l[S_{k-1}']}\left(S_{k-1}',A_{k-1}' \right)=|b_1|.
\end{equation}

Due to Eq. (\ref{Dk}, \ref{Dk-1}), we obtain 

 %By Lemma \ref{CorrectStep}, we have $CorrectionNum_{k-1}=b_3-b_2$. 

\begin{align}
D_{l[S_{k-1}']}\left(S_{k-1}',A_{k-1}' \right)=|b_1|\\
=D_{l[S_{k}']}\left(S_{k}',A_{k}' \right).
\end{align}

By recursion, it can be similarly  proved that
\begin{align}
D_{l[S_{1}']}\left(S_{1}',A_{1}' \right)=|b_1|\\
=D_{l[S_{k}']}\left(S_{k}',A_{k}' \right),
\end{align}
where $S_1'=S'$ and $A_1'=\omega$.

Therefore, {\bf Theorem \ref{ERT}} has been proved.

\section{Application to Distributed Phase Estimation}\label{sec:application}

The error correction algorithm can be applied to designing a distributed phase estimation algorithm. 
Here,  phase estimation algorithm is from \cite{NC2000}, but we will reformulate it with some new notations.	We first recall quantum Fourier transform that is a unitary operator acting on the standard basis states:
\begin{equation}
QFT |j\rangle=\frac{1}{\sqrt{2^n}}\sum_{k=0}^{2^n-1}e^{2\pi ijk/2^n}|k\rangle\text{,}
\end{equation}
for $j=0,1,\cdots,2^n-1$. The inverse quantum Fourier transform is defined as:
\begin{equation}\label{inverse_QFT}
QFT^{-1} \frac{1}{\sqrt{2^n}}\sum_{k=0}^{2^n-1}e^{2\pi ijk/2^n}|k\rangle=|j\rangle\text{,}
\end{equation}
for $j=0,1,\cdots,2^n-1$.
%The quantum Fourier transform and its inverse can be implemented by using $O(n^2)$ elementary gates (i.e., $O(n^2)$ single-qubit and two-qubit gates) \cite{NC2000}.

	Phase estimation algorithm is a practical application of quantum Fourier transform. Let a unitary operator $U$ together with its eigenvector  $|u\rangle$ satisfy
\begin{equation}\label{eq:Uu1}
U|u\rangle=e^{2\pi i\omega}|u\rangle
\end{equation}
 for some real number $\omega\in[0,1)$,  where  $\omega=0.a_1a_2\cdots a_n\cdots$,  $a_i\in \{0,1\}$ for each $i$.  Suppose that the controlled operation $C_m(U)$ is defined as
 \begin{equation}\label{CmU}
C_m(U)|j \rangle|u\rangle=|j\rangle U^j|u\rangle
\end{equation}
for any positive integer $m$ and $m$-bit string $j$, where the first register is control qubits.  %Then we can apply the phase estimation algorithm to estimate $\omega$ (see {\bf Algorithm 2}). 

\begin{remark}\label{remark:PEVariant}
Let $x$ be a natural number. By Eq. (\ref{eq:Uu1}), we have $U^{2^{x-1}}|u\rangle=e^{2\pi i(2^{x-1}\omega)}|u\rangle=e^{2\pi i  0.a_xa_{x+1}\cdots}|u\rangle$. Thus, to estimate $ 0.a_xa_{x+1}\cdots$, we can apply the phase estimation algorithm similarly and change $C_t(U)$  to $C_t(U^{2^{x-1}})$ accordingly.% \cite{li2017application}.
\end{remark}

%The  details of phase estimation algorithm \cite{NC2000} is reviewed in Appendix, but for the convenience of readability, we outline the algorithm and related results. 

Phase estimation algorithm ({\bf Algorithm 2}) \cite{NC2000}  can be reformulated as follows. %to estimate $\omega$ (see {\bf Algorithm 2}). 

%$QFT^{-1}\dfrac{1}{\sqrt{2^t}}\sum\limits_{j=0}^{2^t-1}e^{2\pi ij\omega}|j\rangle |u\rangle=\dfrac{1}{\sqrt{2^t}}\sum\limits_{j=0}^{2^t-1}e^{2\pi ij\omega}QFT^{-1}|j\rangle |u\rangle=\dfrac{1}{\sqrt{2^t}}\sum\limits_{j=0}^{2^t-1}e^{2\pi ij\omega}\dfrac{1}{\sqrt{2^t}}\sum\limits_{k=0}^{2^t-1}e^{-2\pi ijk/2^t}|k\rangle |u\rangle=\dfrac{1}{2^t}\sum\limits_{j=0}^{2^t-1}\sum\limits_{k=0}^{2^t-1}e^{2\pi ij(\omega-k/2^t)}|k\rangle |u\rangle$

%\begin{figure*}
%\begin{minipage}{\linewidth}
%\begin{algorithm}[H]\label{alg:PE}
%\caption{Phase estimation algorithm}

%\textbf{Input}: The probability of success $1-\epsilon$ and initial state $|0\rangle^{\otimes t}|u\rangle$, where $t=n+\lceil\log_2(2+\dfrac{1}{2\epsilon})\rceil$.\\
%\textbf{Output}: An estimation $\widetilde{\omega}$ of $\omega$ with $|\widetilde{\omega}-\omega|<\dfrac{1}{2^n}$.\\
\textbf{Input}: Unitary operator $U$ together with its eigenvector  $|u\rangle$ satisfies $U|u\rangle=e^{2\pi i\omega}|u\rangle$, $n$, and $\epsilon\in (0,1)$.% where  $\omega=0.a_1a_2\cdots a_n\cdots$,  $a_i\in \{0,1\}$ for each $i$.
%A positive integer $n$ (it means that we want to estimate the first $n$ bits of $\omega$) and the success probability $1-\epsilon$ ($\epsilon\in(0,1)$).

\textbf{Output}: A $t$-bit string $\widetilde{\omega}$ satisfies:
 $D_n(P[\widetilde{\omega},n],P[\overline{\omega}, n] )\leq 1$,  $t=n+\lceil\log_2(2+\dfrac{1}{2\epsilon})\rceil$, where $\overline{\omega}=a_1a_2\cdots a_n\cdots$, if $ \omega=0.a_1a_2\cdots a_n\cdots$.
%$d_n(\widetilde{\omega}_{[1,n]},\omega_{\{1,n\}})\leq 1$.

%\textbf{Procedure}:
\begin{algorithmic}[1]
\State Create initial state $|0\rangle|u\rangle$: 
The first register is $t$-qubit.

\State Apply $H^{\otimes t}$ to the first register: 
  \quad  $\dfrac{1}{\sqrt{2^t}}\sum\limits_{j=0}^{2^t-1}|j\rangle|u\rangle$.%$H^{\otimes t}|0\rangle|u\rangle=\dfrac{1}{\sqrt{2^t}}\sum\limits_{j=0}^{2^t-1}|j\rangle|u\rangle$.
\State Apply $C_t(U)$: 
   \quad   $\dfrac{1}{\sqrt{2^t}}\sum\limits_{j=0}^{2^t-1}|j\rangle e^{2\pi ij\omega}|u\rangle$.  
  %$C_t(U)\dfrac{1}{\sqrt{2^t}}\sum\limits_{j=0}^{2^t-1}|j\rangle|u\rangle=\dfrac{1}{\sqrt{2^t}}\sum\limits_{j=0}^{2^t-1}|j\rangle U^j|u\rangle=\dfrac{1}{\sqrt{2^t}}\sum\limits_{j=0}^{2^t-1}|j\rangle e^{2\pi ij\omega}|u\rangle$.
\State Apply $QFT^{-1}$: 
 \quad $\dfrac{1}{2^t}\sum\limits_{j=0}^{2^t-1}\sum\limits_{k=0}^{2^t-1}e^{2\pi ij(\omega-k/2^t)}|k\rangle |u\rangle$. 
 
 %$QFT^{-1}\dfrac{1}{\sqrt{2^t}}\sum\limits_{j=0}^{2^t-1}e^{2\pi ij\omega}|j\rangle |u\rangle=\dfrac{1}{2^t}\sum\limits_{j=0}^{2^t-1}\sum\limits_{k=0}^{2^t-1}e^{2\pi ij(\omega-k/2^t)}|k\rangle |u\rangle$. %|PE_{t,\omega}\rangle|u\rangle$.
 
\State Measure the first register: 
 \quad obtain a $t$-bit string $\widetilde{\omega}$.
\end{algorithmic}

The above phase estimation algorithm is used to estimate $\omega$, which can be more accurately described by the following propositions.

%\begin{proposition}[See \cite{nielsen2000quantum}]\label{PE}
	%In {\bf Algorithm 2}, for any $\epsilon>0$ and any positive integer $n$, if $t=n+\lceil\log_2(2+\dfrac{1}{2\epsilon})\rceil$, then the probability of $D_t(\widetilde{\omega},P[\overline{\omega}, t] )<2^{t-n}$ is at least $1-\epsilon$, where $\overline{\omega}=a_1a_2\cdots a_n\cdots$, 	if $\omega=0.a_1a_2\cdots a_n\cdots$.
	%\end{proposition}

	\begin{proposition}[See \cite{NC2000}]\label{PE}
	In {\bf Algorithm 2},  if $t=n+\lceil\log_2(2+\dfrac{1}{2\epsilon})\rceil$, then the probability of $D_t(\widetilde{\omega},P[\overline{\omega}, t] )<2^{t-n}$ is at least $1-\epsilon$. %where $\overline{\omega}=a_1a_2\cdots a_n\cdots$, 	if $\omega=0.a_1a_2\cdots a_n\cdots$.
	\end{proposition}

Due to Proposition \ref{PE} and Proposition \ref{d_t} (III) we have the following result.

%\begin{proposition}\label{PE2}
	%In {\bf Algorithm 2}, for any $\epsilon>0$ and any positive integer $n$, if $t=n+\lceil\log_2(2+\dfrac{1}{2\epsilon})\rceil$, then the probability of $D_n(P[\widetilde{\omega},n],P[\overline{\omega}, n] )\leq 1$ is at least $1-\epsilon$. 
%\end{proposition}

\begin{proposition}\label{PE2}
	In {\bf Algorithm 2},  if $t=n+\lceil\log_2(2+\dfrac{1}{2\epsilon})\rceil$, then the probability of $D_n(P[\widetilde{\omega},n],P[\overline{\omega}, n] )\leq 1$ is at least $1-\epsilon$. 
\end{proposition}

%\subsection{Distributed phase estimation algorithm}

%\begin{remark}
%If $1\leq d(P[\overline{\omega}, n])< 2^n-1$, then $|P[\widetilde{\omega},n],P[\overline{\omega}, n] |\leq 1$ holds from $D_n(P[\widetilde{\omega},n],P[\overline{\omega}, n] )\leq 1$.
%\end{remark}

%	By means of   Proposition \ref{PE2}, the following corollary can be obtained.

%\begin{corollary}\label{MeasurePE}
%Let $n$ be a positive integer and let $\omega\in[0,1),\epsilon\in(0,1),t=n+\lceil\log_2(2+\dfrac{1}{2\epsilon})\rceil$. Denote $M=\{x\in\{0,1\}^n: D_n(x,P[\overline{\omega}, n])\leq 1\}$ and $P_M=\sum_{a\in M}|a\rangle\langle a|$.  Then
%\begin{equation}
%\left\|P_MQFT^{-1}\dfrac{1}{\sqrt{2^t}}\sum\limits_{j=0}^{2^t-1}e^{2\pi ij\omega}|j\rangle \right\|^2\geq 1-\epsilon.
%\end{equation}
%\end{corollary}

%\begin{corollary}\label{MeasurePE}
%In {\bf Algorithm 2},  let $t=n+\lceil\log_2(2+\dfrac{1}{2\epsilon})\rceil$. Denote $M=\{x\in\{0,1\}^n: D_n(x,P[\overline{\omega}, n])\leq 1\}$ and $P_M=\sum_{a\in M}|a\rangle\langle a|$.  Then
%\begin{equation}
%\left\|P_MQFT^{-1}\dfrac{1}{\sqrt{2^t}}\sum\limits_{j=0}^{2^t-1}e^{2\pi ij\omega}|j\rangle \right\|^2\geq 1-\epsilon.
%\end{equation}
%\end{corollary}

%By Proposition \ref{d_t} (III), we can derive  $D_n(P[\widetilde{\omega},n],P[\overline{\omega}, n] )\leq 1$	from  $D_t(\widetilde{\omega},P[\overline{\omega}, t] )\leq 1$.  

Proposition \ref{PE2} implies that it  requires $\lceil\log_2(2+\dfrac{1}{2\epsilon})\rceil$ additional qubits for estimating the first $n$ bits of $\omega$ with success probability at least $1-\epsilon$ and with deviation of error no larger than $1$. %we need to measure $\lceil\log_2(2+\dfrac{1}{2\epsilon})\rceil$ additional qubits.

As reviewed  above,	phase estimation algorithm can estimate  $P[\overline{\omega}, n]=a_1a_2\cdots a_n$ of $\omega$ in Eq. (\ref{eq:Uu1}) for given $n$. Then we apply {\bf Error Correction Algorithm} to designing a distributed phase estimation algorithm. 	If we design $k$ computing nodes, then   $\omega=a_1a_2\ldots a_n\in\{0,1\}^n$ is divided into $k$ substrings,  say $A_1,A_2,\ldots,A_k$,  and for a given $k_0$ (the length of each substring $A_i$ is not smaller than $k_0$),  %the suffix of $A_i$ with length $3$ is overlapped with the prefix of $A_{i+1}$ with the same length $3$,  
 $S[A_i, k_0]=P[A_{i+1},k_0]$,
 $i=1,2,\ldots,k-1$. 
 
 In the interest of simplicity, and also without loss of generality, we take $k_0=3$, all substrings $A_i$ ($i=1,2,\ldots,k-1$)   have the same length $N_0\geq 3$. %and  the length of  $A_k$ is $n-(k-1)(N_0-3)$. %where we also require $3\leq n-(k-1)(N_0-3)\leq N_0$. 
  It is easy to see that the subscript of first bit of $A_i$ is $(i-1)N_0-3(i-1)+1$, denoted by $l_i$ for short.

We outline the basic idea of our distributed phase erstimation algorithm. $k$ computing nodes are denoted as $Q_1,\cdots,Q_k$ to estimate  $A_1,A_2,\ldots, A_k$, respectively. In fact, we also employ phase estimation algorithm for each computing node by adjusting some parameters appropriately. If $A_i$ ($i=1,2,\ldots,k-1$) is estimated, then $n$, $t$, and $C_{t}(U) $ are 
   replaced by $N_0$,  $t_i=N_0+\lceil\log_2(2+\dfrac{k}{2\epsilon})\rceil$, and $C_{t_i}(U^{2^{l_i-1}})$, respectively. For estimating $A_k$, we need to use the length of $A_k$ to replace $N_0$.

    For each computing node, measurement is performed to the first $t_i$ bits and then its prefix with length $N_0$ denoted by $S_i$ is achieved as the estimation of $A_i$.  Finally, we apply error correction algorithm to these $k$ estimation values $S_1,S_2,\cdots, S_k$ and obtain an estimation  of phase $\omega$.
   
   So, we give a distributed phase estimation algorithm ({\bf Algorithm 3})  as follows.
   
   \textbf{Input}:  Unitary operator $U$ together with its eigenvector  $|u\rangle$ satisfies $U|u\rangle=e^{2\pi i\omega}|u\rangle$, $\epsilon\in(0,1)$;  $n$, $N_0\geq 3$,  and $k$ satisfy $3\leq n-(k-1)(N_0-3)=l[A_k]\leq N_0$; $l_i=(i-1)N_0-3(i-1)+1$.
   
%A positive integer $n$ (it means that we want to estimate the first $n$ bits of $\omega$) and the success probability $1-\epsilon$ ($\epsilon\in(0,1)$).\\
\textbf{Output}: An $n$-bit string $S'$ such that $D_n(S',P[\overline{\omega}, n])\leq 1$ with success probability at least $1-\epsilon$.

%\textbf{Procedure}:\\
Nodes $Q_1,Q_2,\cdots,Q_k$ perform the following operations in parallel.
\quad\textbf{ Node $Q_i$ excute ($i=1,2,\cdots,k$)}:
\begin{algorithmic}[1]

\State Create initial state $|0\rangle_{R_i}|u\rangle$:
Register $R_i$ is $t_i$-qubit, where $t_i=N_0+\lceil\log_2(2+\dfrac{k}{2\epsilon})\rceil$ for $i=1,2,\cdots,k-1$, and $t_k=l(A_k)+\lceil\log_2(2+\dfrac{k}{2\epsilon})\rceil$.
\State Apply $H^{\otimes t_i}$ to the first register:
  \quad. $\dfrac{1}{\sqrt{2^{t_i}}}\sum\limits_{j=0}^{2^{t_i}-1}|j\rangle_{R_i}|u\rangle$
  \State Apply $C_{t_i}(U^{2^{l_i-1}})$: %$H^{\otimes t_i}|0\rangle_{R_i}|u\rangle=\dfrac{1}{\sqrt{2^{t_i}}}\sum\limits_{j=0}^{2^{t_i}-1}|j\rangle_{R_i}|u\rangle$. :
\quad %$C_{t_i}(U^{2^{l_i-1}})\dfrac{1}{\sqrt{2^{t_i}}}\sum\limits_{j=0}^{2^{t_i}-1}|j\rangle_{R_i}|u\rangle=$
%$\dfrac{1}{\sqrt{2^{t_i}}}\sum\limits_{j=0}^{2^{t_i}-1}|j\rangle_{R_i} (U^{2^{l_i-1}})^j|u\rangle$
%=$\dfrac{1}{\sqrt{2^{t_i}}}\sum\limits_{j=0}^{2^{t_i}-1}|j\rangle_{R_i} e^{2\pi ij(2^{l_i-1}\omega)}|u\rangle=$
$\dfrac{1}{\sqrt{2^{t_i}}}\sum\limits_{j=0}^{2^{t_i}-1}|j\rangle_{R_i} e^{2\pi ij0.a_{l_i}a_{l_i+1}\cdots }|u\rangle$.
\State Apply $QFT^{-1}$:
% $QFT^{-1}\dfrac{1}{\sqrt{2^{t_i}}}\sum\limits_{j=0}^{2^{t_i}-1}e^{2\pi ij0.a_{l_i}a_{l_i+1}\cdots } |j\rangle_{R_i} |u\rangle=$

   $\dfrac{1}{2^{t_i}}\sum\limits_{j=0}^{2^{t_i}-1}\sum\limits_{k=0}^{2^{t_i}-1}e^{2\pi ij(0.a_{l_i}a_{l_i+1}\cdots-k/2^t)}|k\rangle |u\rangle.$ %$|PE_{t_r,\omega_{\{l_r,+\infty\}}}\rangle_{R_r}|u\rangle$.\\
 \State Measure the  first register $R_i$ and obtain a $t_i$ bits string: 
 denote  its prefix with length $N_0$ as $S_i$ for $i=1,2,\cdots,k-1$, and  its prefix with length $l[A_k]$ as $S_k$.
 \State  \textbf{Execute ``Error Correction Algorithm" with inputting $S_1,\cdots,S_k$, and output an $n$-bit string $S'$.}
%\quad 
%\State $S'\leftarrow ErrorCorrection(S_1,\cdots,S_k)$:
%\quad $S'$ is an $n$-bit string.
\State Return $S'$.
   \end{algorithmic}

  By means of Proposition  \ref{PE2}, we have the following corollary straightforward.
\begin{corollary}\label{DPE2}
	In {\bf Algorithm 3},  %given $t_i=N_0+\lceil\log_2(2+\dfrac{k}{2\epsilon})\rceil$,  
	the probability of $D_{l[A_i]}(S_i,A_i)\leq 1$ is at least $1-\epsilon$,  $i=1,2,\cdots,k.$
	\end{corollary}

	%Since $A_i=P[\overline{\omega}_{l_i}, N_0]$, we have  the probability of $D_{N_0}(S_i,A_i )\leq 1$ is at least $1-\epsilon$, for  $i=1,2,\cdots,k.$	
	Due to Theorem \ref{ERT}, we have the following result.

\begin{theorem}\label{ECDPE}
In Algorithm 3, the probability of 
\begin{equation}
D_n(S',P[\overline{\omega}, n])\leq 1
\end{equation}
is at least $1-\epsilon$. 
\end{theorem}

\begin{proof}
According to  Theorem \ref{ERT},  we know that 
\begin{equation}
D_n(S',P[\overline{\omega}, n])=D_{l(S_k)}(S_k,A_k).
\end{equation}
By combining Corollary \ref{DPE2}, the theorem follows.

\end{proof}

	Finally, we analyze the complexity of the above distributed phase estimation algorithm. As we know, in phase estimation algorithm,  the main operator $C_t(U)$   can be implemented by $t$ controlled operators in the form of controlled-$U^{2^x} $\cite{NC2000}, $x=0,1,2,3,4,\cdots,t-1$. Therefore,  the number of  controlled-$U^{2^x}$ gates is taken as a metric in the complexity analysis. 
	
	In {\bf Algorithm 3}, the qubits  and the number of controlled-$U^{2^x}$ of per node are $\dfrac{n}{k}+\log_2k+N(|u\rangle)+O(1)$ and  $\dfrac{n}{k}+\log_2k+O(1)$, respectively, where $N(|u\rangle)$ denotes the number of  qubit of $|u\rangle$.	 Our distributed phase estimation algorithm does not require quantum communication. Compared with the centralized phase estimation algorithm, the maximum number of qubits required by a single computing node ($Q_i$) in our distributed algorithm is reduced by $(1-\dfrac{1}{k})n-\log_2k-O(1)$.

\section{Conclusions} \label{sec:conclusions}

 In this paper, we have proposed a universal method  of error correction for a kind of distributed quantum computing, and then we have applied this method to designing a distributed phase estimation algorithm.  In general, if the solution of a problem can be represented as a bit string, and there are multiple computing nodes to obtain respective substrings approximately, then  the error correction scheme in this paper can be used to achieve an approximate solution efficiently.
 
 Phase estimation  is a basic algorithm that can be used to design other quantum algorithms, so naturally,  the distributed phase estimation algorithm in this paper can  be  used to design other distributed quantum algorithms, for example, distributed order-finding algorithm, distributed factoring algorithm, distributed discrete logarithm algorithm, and distributed HHL algorithm  et al \cite{NC2000,Shor1994,HHL2009}.

\newpage

\appendix

\end{document}